\documentclass[submission,copyright,creativecommons]{eptcs}
\providecommand{\event}{M4M 2016} 
\usepackage{breakurl}             
\usepackage{underscore}           

\title{Strong Completeness and the Finite Model Property for Bi-Intuitionistic Stable Tense Logics}
\author{Katsuhiko Sano
\institute{Department of Philosophy, Hokkaido University, Hokkaido, Japan}
\email{v-sano@let.hokudai.ac.jp}
\and
John G. Stell 
\institute{School of Computing, University of Leeds, Leeds, UK}
\email{\quad J.G.Stell@leeds.ac.uk}
}
\def\titlerunning{Strong Completeness and FMP for Bi-intuitionistic Stable Tense Logics}
\def\authorrunning{K. Sano and J.~G.~Stell}

\usepackage{leeds2016m4m_rev}

\begin{document}

\maketitle

\begin{abstract}
Bi-Intuitionistic Stable Tense Logics (BIST Logics) are tense logics with
a Kripke semantics where worlds in a frame are equipped with a pre-order
as well as with an accessibility relation which is `stable' with respect 
to this pre-order. BIST logics are extensions of a logic, $\mathbf{BiSKt}$, which 
arose in the semantic context of hypergraphs, since a special case of the 
pre-order can represent the incidence structure of a hypergraph. In this 
paper we provide, for the first time, a Hilbert-style axiomatisation of 
BISKt and prove the strong completeness of $\mathbf{BiSKt}$. We go on to prove strong 
completeness of a class of BIST logics obtained by extending $\mathbf{BiSKt}$ by 
formulas of a certain form. Moreover we show that the finite model property and the decidability hold for a class of BIST logics. 
\end{abstract}

\section{Introduction}

The motivation for the logics in this paper comes from mathematical morphology as used in image processing and from the extension of this body of techniques 
to the case of sets of pixels having the structure of a graph.
A black and white image can be identified with a subset of $\mathbb{Z}^2$; the black pixels being those in the subset.
Mathematical morphology uses what it calls a {\em structuring element} as a `probe' or a `lens' through which to view an image. By use of the appropriate structuring element
certain features of an image may be removed or rendered less prominent whereas other features can be accentuated. The mathematical basis of this approach is that
a relation $R$ on a set $U$ provides functions which map subsets of $U$ to subsets of $U$. In mathematical morphology these functions are known as {\em dilation} and {\em erosion},
defined as follows where $X \subseteq U$.
$$\begin{array}{lrcl}
\text{Dilation: } & X \dilate  R & = & \{u \in U \mid \text{ for some $x$} \; (x \mathrel{R} u \text{ and } x \in X)\},\\[1ex]
\text{Erosion: }  & R \erode   X & = & \{u \in U \mid \text{ for all $x$} \; (u \mathrel{R} x \text{ implies }    x \in X)\}.  
\end{array}$$
From these operations the {\em opening} and {\em closing} of $X$ by $R$ are defined respectively by $(R \erode X) \dilate R$ and 
$R \erode (X \dilate R)$. Given an appropriate choice of $R$ and when $X$ consists of the black pixels in an image,
 the opening can be used to remove unwanted black pixels from the image. Dually the the closing can be used to add pixels to an image, for example to fill in white cracks between black parts of an image. 
 The most common formulation of mathematical morphology is in terms of structuring elements which usually consist of small patterns of pixels, however it is
 well-known~\cite[p.60]{NajmanTalbot2010} that every structuring element gives rise to a relation on the set of all pixels.
 
A recent development in mathematical morphology has been the extension of the techniques to graphs~\cite{CoustyNajmanCVIU2013}.
One possibility is to start with $\mathbb{Z}^2$ as a grid of pixels but with the graph structure that arises by taking pixels as nodes and putting edges between adjacent pixels. A subgraph then specifies a set of pixels but not necessarily containing all the adjacencies (edges) from the underlying graph.
To extend the operations of dilation and erosion to the case of graphs needs an appropriate notion of a relation on a graph.
While there are several possible notions of what should be meant by a relation on a graph, the definition which appears in~\cite{Stell2015} can be justified in view
of the bijective correspondence between these relations and union-preserving functions on the lattice of subgraphs. Relations on graphs are most conveniently developed in the more general situation of hypergraphs -- in which edges can be incident with arbitrary non-empty sets of nodes. These relations are relations on the set of all edges
and nodes which statisfy a stability condition.

\begin{definition}
A {\em hypergraph} is a set $U$ together with an incidence relation $H \subseteq U \times U$ such that $H$ is reflexive and whenever 
$x \mathrel{H} y$ and $y \mathrel{H} z$ then $x = y $ or $y = z$.
A hypergraph is thus a special case of a set $U$ equipped with a pre-order, $H$. A subset $X \subseteq U$ of a hypergraph $(U,H)$ is a {\em sub hypergraph} or more briefly a {\em subgraph} when for all $x \in X$, if $x \mathrel{H} u$ then
$u \in X$.
\end{definition}
Given a hypergraph, $(U, H)$ every $u \in U$ is either an {\em edge} or a {\em node}. The nodes are those elements, $u$,  for which $u \mathrel{H} v$ implies $u = v$,
and the edges are those $u$ for which there exists a $v$ such that $u \mathrel{H} v$ and $u \neq v$.
If $u$ is an edge, and $v$ is a node, and $u \mathrel{H} v$ then we say that $u$ and $v$ are {\em incident}.
A graph is the special case of a hypergraph in which each edge is incident with either one or two nodes.
The relations we consider on hypergraphs  are
those $R \subseteq U \times U$ which satisfy the stability condition $H;R;H \subseteq R$.

Connections between modal logic and mathematical morphology have already been studied by~\cite{BlochJANCL2002} in the set-based case.
While not following exactly the approach in~\cite{BlochJANCL2002}, the general idea is that 
a frame $(U,R)$ in Kripke semantics provides a set $U$ of pixels, subsets of which are black pixels in  particular images,
and an accessibility relation $R$ generated by a structuring element.
The modalities $\dia$, $\bdia$ and $\Box$ function semantically as
operators taking subsets to subsets, with $\dia$  
associated to $X \mapsto X \dilate \breve{R}$ (where $\breve{R}$ is the ordinary converse of $R$), $\bdia$ associated to $X \mapsto X \dilate R$, and $\Box$ associated to $X \mapsto R \erode X$. A propositional variable $p$ can be understood as a set of pixels; 
the opening and closing of this image are then $\bdia \Box p$ and $\Box \bdia p$ respectively,
and properties of the morphological operations and statements about specific images and parts of images can be expressed in the logic.

Set-based mathematical morphology is naturally associated to a logic which is classical (since the set of subsets of a set is
a Boolean algebra)  and which is modal (since each structuring element provides an accessibility relation).
The generalization of this to mathematical morphology on hypergraphs results in a logic which is bi-intuitionistic (since the set
of subgraphs of a hypergraph forms a bi-Heyting algebra) and which is modal too. 
These considerations led to the development of a bi-intititionistic tense logic, $\mathbf{BiSKt}$, in~\cite{Stell2016}.

The study of intuitionistic modal and tense logic has been done in the several literature, e.g., in~\cite{Esakia2006,Ewald1986,Ono1977,Wolter1999,Sotirov1980,Hasimoto2001}. The bi-intuitionistic tense logic studied in~\cite{Stell2016} can be regarded as both a tense expansion of bi-intuitionistic logic and an expansion of intuitionistic modal logic studied in~\cite{Ono1977,Wolter1999} with the `past-tense' operator and the coimplication. Hilbert-system of bi-intuitionistic logic was first given by Rauszer~\cite{Rauszer1974a}. When we pay attention to the intuitionistic modal fragment of bi-intuitionistic stable tense logic, our Kripke semantics coincides with Kripke semantics for intuitionistic modal logic given in~\cite{Wolter1999}. Finite model property of intuitionistic modal logic was studied in~\cite{Ono1977,Sotirov1980,Hasimoto2001}. As far as the authors know, any general results on strong completeness and finite model property have not been yet known for bi-intuitionistic tense logic. This paper gives the first step toward this direction (see also the discussion in Section \ref{sec:Related}). 

This paper is organized as follows. Section \ref{sec:semantics} introduces the notion of stable relations on preorders and then moves to our syntax and Kripke semantics of bi-intuitionistic stable tense logic. Section \ref{sec:axiomatisation} provides Hilbert system for the smallest bi-intuitionistic stable tense logic and shows that it is sound for Kripke semantics. We note that our axiomatisation for bi-intuitionistic fragment is much simpler than Rauszer's axiomatisation~\cite{Rauszer1974a}. Section \ref{sec:completeness} establishes the strong completeness of the bi-intuitionistic stable tense logic for Kripke semantics and Section \ref{sec:extension} extends this argument to several extensions of $\mathbf{BiSKt}$. Finally, Section \ref{sec:fmp} follows Hasimoto's technique for intuitionistic modal logic to show the finite model property of some extensions of $\mathbf{BiSKt}$, which implies the decidability of the extensions, provided the logics are finitely axiomatizable. Section \ref{sec:Related} comments on the related literature and future work of this paper.

\section{Kripke Semantics for Bi-intuitionistic Stable Tense Logic}
\label{sec:semantics}
\subsection{Stable Relations on Preorders}

\begin{definition}
\label{defn-stable}
Let $H$ be a preorder on a set $U$. We say that $X \subseteq U$ is an {\em $H$-set} if $X$ is closed under $H$-successors, i.e., $uHv$ and $u \in X$ jointly imply $v \in X$ for all elements $u$, $v \in U$. 
Given a preorder $(U,H)$, a binary relation $R \subseteq U \times U$ is {\em stable} if it satisfies $H;R;H \subseteq R$.
\end{definition}

\noindent It is easy to see that a relation $R$ on $U$ is stable, if and only if, $R; H \subseteq R$ and $H;R \subseteq R$. Given any binary relation $R$ on $U$,  $\breve{R}$ is defined as the ordinary converse of $R$. Even if $R$ is a stable relation on $U$, its converse $\breve{R}$ may be not stable. 

\begin{definition}[\cite{Stell2016}]
The {\em left converse} $\leftconv R$ of a stable relation  $R$ is $H;{\breve{R}};H$. 
\end{definition}

It is easy to verify the stability of the left converse $\leftconv R$. 

\begin{example}
\label{ex:stable}
\begin{enumerate}
\item A preorder $(U_{1},H_{1})$ is defined as follows: $U_{1}$ = $\setof{0,1,2,3}$ and $H_{1}$ = $\setof{(0,1), (2,3)} \cup \inset{(u,u)}{u \in U_{1}}$. Take a stable relation $R_{1}$ on $(U_{1},H_{1})$ defined by $R_{1}$ = $\setof{(0,3),(1,3)}$ (see the first graph below, where the double solid lines are for $R_{1}$ and the single solid lines are for $H_{1}$ but the reflexive $H_{1}$-arrows are disregarded). Then $\breve{R_{1}}$ = $\setof{(3,0), (3,1)}$ is not stable, but $\leftconv R_{1}$ = $\setof{(3,0),(3,1),(2,0),(2,1)}$ is stable (see the dotted double lines in the second graph below). 
\item A preorder $(U_{2},H_{2})$ is defined as follows: $U_{2}$ = $\setof{a,b,c}$ and $H_{2}$ = $\setof{(a,b)} \cup \inset{(u,u)}{u \in U_{2}}$. Take a stable relation $R_{2}$ on $(U_{2},H_{2})$ defined by $R_{2}$ = $\setof{(a,c),(b,c)}$ (see the third graph below where we follow the same convension as in item (1)). Then $\breve{R_{2}}$ = $\setof{(c,a), (c,b)}$ is already stable and so $\breve{R_{2}}$ = $\leftconv R_{2}$ (see the dotted double lines in the fourth graph below). 
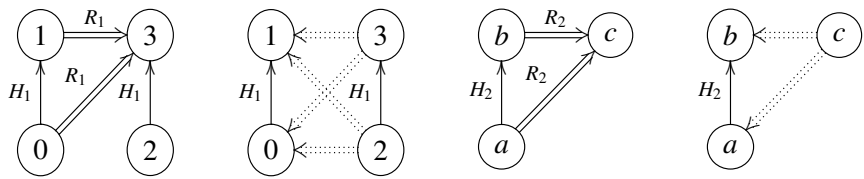
\[
\xymatrix{
*++[o][F-]{1} \ar@{=>}[r]^{R_{1}} & *++[o][F-]{3} \\
*++[o][F-]{0} \ar[u]^{H_{1}} \ar@{=>}[ur]^{R_{1}}  & *++[o][F-]{2} \ar[u]^{H_{1}} \\
}\qquad
\xymatrix{
*++[o][F-]{1}  & *++[o][F-]{3} \ar@{:>}[l] \ar@{:>}[dl] \\
*++[o][F-]{0} \ar[u]^{H_{1}}  & *++[o][F-]{2} \ar[u]^{H_{1}} \ar@{:>}[l] \ar@{:>}[ul] \\
}
\qquad
\xymatrix{
*++[o][F-]{b} \ar@{=>}[r]^{R_{2}} & *++[o][F-]{c} \\
*++[o][F-]{a} \ar[u]^{H_{2}} \ar@{=>}[ur]^{R_{2}}  & \\
}
\qquad
\xymatrix{
*++[o][F-]{b}  & *++[o][F-]{c} \ar@{:>}[l] \ar@{:>}[dl] \\
*++[o][F-]{a} \ar[u]^{H_{2}}  &  \\
}
\]
\end{enumerate}
\end{example}

\subsection{Syntax and Kripke Semantics for Bi-intuitionistic Stable Tense Logic}

Let $\mathsf{Prop}$ be a countable set of propositional variables. Our syntax $\mathcal{L}(\mathsf{Mod})$ for bi-intuitionistic stable tense logic consists of all logical connectives of bi-intuitionistic logic, i.e., two constant symbols  $\bot$ and $\top$, disjunction $\lor$, conjunction $\land$, implication $\to$, coimplication $\coimp$, and a finite set $\mathsf{Mod} \subseteq \setof{\bdia,\Box, \dia, \bbox}$ of modal operators. The set of all formulas in $\mathcal{L}(\mathsf{Mod})$ is defined in a standard way. For example, when $\mathsf{Mod}$ = $\setof{\bdia,\Box, \dia, \bbox}$, the set $\mathsf{Form}_{\mathcal{L}(\mathsf{Mod})}$ of all formulas of the syntax $\mathcal{L}(\mathsf{Mod})$ is defined inductively as follows:
\[
\varphi ::= \top \,|\,  \bot \,|\, p \,|\, \varphi \land \varphi  \,|\,\varphi \lor \varphi \,|\, \varphi \to \varphi \,|\, \varphi \coimp \varphi \,|\,  \bdia \varphi \,|\, \Box \varphi \,|\,  \dia \varphi \,|\, \bbox \varphi \quad (p \in \mathsf{Prop}).
\]
In this section, we mostly assume that $\mathsf{Mod}$ = $\setof{\bdia,\Box, \dia, \bbox}$. 
We define the following abbreviations: 
\[
\neg \varphi := \varphi \to \bot, \quad \coneg \varphi :=  \top \coimp \varphi, \quad \varphi \leftrightarrow \psi := (\varphi \to \psi) \land (\psi \to \varphi).
\]
\begin{definition}
\label{defn-semantics}
$F$ = $(U,H,R)$ is an {\em $H$-frame} if $U$ is a nonempty set, $H$ is a preorder on $U$, and $R$ is a {\em stable} binary relation on $U$. A {\em valuation} on an $H$-frame $F$ = $(U,H,R)$ is a mapping $V$ from $\mathsf{Prop}$ to the set of all $H$-sets on $U$. $M$ = $(F,V)$ is an {\em $H$-model} if $F$ = $(U,H,R)$ is an $H$-frame and $V$ is a valuation. 
Given an $H$-model $M$ = $(U,H,R,V)$, a state $u \in U$ and a formula $\varphi$, the satisfaction relation $M,u \models \varphi$ is defined inductively as follows:
\[
\begin{array}{lll}
M,u \models p &\iff& u \in V(p), \\
M,u \models \top, && \\
M,u \not\models \bot, & & \\
M,u \models  \varphi \lor \psi &\iff& M,u \models  \varphi \text{ or } M,u \models \psi,  \\
M,u \models  \varphi \land \psi &\iff& M,u \models  \varphi \text{ and } M,u \models \psi,  \\
M,u \models  \varphi \to \psi &\iff& \text{ For all $v \in U$ } ((uHv \text{ and } M,v \models \varphi) \text{ imply } M,v \models \psi),\\
M,u \models \varphi \coimp \psi &\iff& \text{ For some $v \in U$ }(vHu \text{ and } M,v\models \varphi \text{ and }M,v \not\models \psi), \\
M,u \models \bdia \varphi &\iff& \text{ For some $v \in U$ } (vRu \text{ and } M,v \models \varphi),\\
M,u \models \Box \varphi &\iff& \text{ For all $v \in U$ } (uRv \text{ implies } M,v \models \varphi),\\
M,u \models \dia \varphi &\iff&  \text{ For some $v \in U$ } ((v,u) \in \leftconv R \text{ and } M,v \models \varphi), \\
M,u \models \bbox \varphi &\iff& \text{ For all $v \in U$ } ((u,v) \in \leftconv R \text{ implies } M,v \models \varphi).\\
\end{array}
\]
The {\em truth set} $\den{\varphi}_{M}$ of a formula $\varphi$ in an $H$-model $M$ is defined by $\den{\varphi}_{M}$ := $\inset{u \in U}{M,u \models \varphi}$. If the underlying model $M$ in $\den{\varphi}_{M}$ is clear from the context, we drop the subscript and simply write $\den{\varphi}$. We write $M \models \varphi$ (read: `$\varphi$ is valid in $M$') to means that $\den{\varphi}_{M}$ = $U$ or $M,u\models \varphi$ for all states $u \in U$. For a set $\Gamma$ of formulas, $M \models \Gamma$ means that $M \models \gamma$ for all $\gamma \in \Gamma$. 
Given any $H$-frame $F$ = $(U,H,R)$, we say that a formula $\varphi$ is {\em valid} in $F$ (written: $F \models \varphi$) if $(F,V)  \models \varphi$ for any valuation $V$ and any state $u \in U$, i.e., $\den{\varphi}_{(F,V)}$ = $U$. 
\end{definition}

As for the abbreviated symbols, we may derive the following satisfaction conditions:
\[
\begin{array}{lll}
M,u \models \neg \varphi &\iff& \text{ For all $v \in U$ } (uHv \text{ implies } M,v \not\models \varphi),\\
M,u \models \coneg \varphi &\iff& \text{ For some $v \in U$ }(vHu \text{ and } M,v \not\models \varphi),\\
M,u \models \varphi \leftrightarrow \psi  &\iff& \text{ For all $v \in U$ } (uHv \text{ implies } (M,v \models \varphi   \iff M,v \models \psi) ). \\
\end{array}
\]

\begin{proposition}[\cite{Stell2016}]
Given any $H$-model $M$, the truth set $\den{\varphi}_{M}$ is an $H$-set.
\end{proposition}

\begin{proof}
By induction on $\varphi$. When $\varphi$ is of the form $\bdia \psi$, $\Box \psi$, $\dia \psi$ or $\bbox \psi$, we need to use $R;H \subseteq R$, $H;R \subseteq R$, $\leftconv R;H \subseteq \leftconv R$, $H;\leftconv R \subseteq \leftconv R$, respectively. Note that these properties hold since $R$ and $\leftconv R$ are stable.
\end{proof}

\begin{definition}
Given a set $\Gamma \cup \setof{\varphi}$ of formulas, $\varphi$ is a {\em semantic consequence} of $\Gamma$ (notation: $\Gamma \models \varphi$) if, whenever $M,u \models \gamma$ for all $\gamma \in \Gamma$, $M,u \models \varphi$ holds, for all $H$-models $M$ = $(U,H,R,V)$ and all states $u \in U$. When $\Gamma$ is a singleton $\setof{\psi}$ of formulas, we simply write $\psi \models \varphi$ instead of $\setof{\psi} \models \varphi$. When both $\varphi \models \psi$ and $\psi \models \varphi$ hold, we use $\varphi \eqv \psi$ to mean that they are equivalent with each other. When $\Gamma$ is an emptyset, we also simply write $\models \varphi$ instead of $\emptyset \models \varphi$. 
\end{definition}

The following proposition is easy to verify (cf.~\cite[Lemma 11]{Stell2016}). 

\begin{proposition}[\cite{Stell2016}]
\label{prop:sem-basic}
\begin{enumerate}
\item $\Gamma \models \varphi \to \psi$ iff $\Gamma \cup \setof{\varphi} \models \psi$. 
\item $F \models \varphi \land \psi \to \gamma$ iff 
$F \models \varphi \to (\psi \to \gamma)$. Therefore, $\varphi \land \psi \models \gamma$ iff $\varphi \models \psi \to \gamma$.
\item $F \models (\varphi \coimp \psi) \to \gamma$ iff $F \models \varphi \to (\psi \lor \gamma)$. Therefore, $\varphi \coimp \psi \models \gamma$ iff $\varphi \models \psi \lor \gamma$.
\item $F \models \bdia \varphi \to \psi$ iff $F \models \varphi \to \Box \psi$. Therefore, $\bdia \varphi \models \psi$ iff $\varphi \models \Box \psi$.
\item $F \models \dia \varphi  \to \psi$ iff $F \models\varphi \to \bbox \psi$. Therefore, $\dia \varphi \models \psi$ iff $\varphi \models \bbox \psi$.
\item $\dia \varphi \eqv \coneg \Box \neg \varphi$.
\item $\bbox \varphi \eqv \neg \bdia \coneg \varphi$.
\end{enumerate}
\end{proposition}

The last two items suggest that $\dia$ and $\bbox$ are definable in the syntax $\mathcal{L}(\bdia,\Box)$. In this sense, we may drop the modal operators $\dia$ and $\bbox$ from the syntax $\mathcal{L}(\bdia,\Box, \dia, \bbox)$. Conversely, we may also ask the question if $\bdia$ and $\Box$ are definable in $\mathcal{L}(\bbox)$ and $\mathcal{L}(\dia)$, respectively. We give negative answers to these two questions. For this purpose, we introduce the appropriate notion of {bounded morphism} for the syntax $\mathcal{L}(\dia,\bbox)$. 

\begin{definition}
Let $M_{i}$ := $(U_{i},H_{i},R_{i},V_{i})$ ($i$ = 1 or 2) be an $H$-model. We say that a mapping $f: U_{1} \to U_{2}$ is a {\em $\mathcal{L}(\dia,\bbox)$-bounded morphism} if it satisfies the following conditions:
\begin{description}
\item[(Atom)] $u \in V_{1}(p)$ $\iff$ $f(u) \in V_{2}(p)$ for all $p \in \mathsf{Prop}$,
\item[($H$-forth)] $uH_{1}v$ implies $f(u)H_{2}f(v)$,
\item[($H$-back)] $f(u)H_{2}v'$ implies that $uH_{1}v$ and $f(v)$ = $v'$ for some $v \in U_{1}$,
\item[($\breve{H}$-back)]  $f(u)\breve{H_{2}}v'$ implies that $u\breve{H_{1}}v$ and $f(v)$ = $v'$ for some $v \in U_{1}$,
\item[($\leftconv R$-forth)] $(u,v) \in \leftconv R_{1}$ implies $(f(u),f(v)) \in \leftconv R_{2}$,
\item[($\leftconv R$-back)] $(f(u),v') \in \leftconv R_{2}$ implies that $(u,v) \in \leftconv R_{1}$ and $f(v)$ = $v'$ for some $v \in U_{1}$,
\item[($\breve{\leftconv R}$-back)]  $(v',f(u)) \in \leftconv R_{2}$ implies that $(v,u) \in \leftconv   R_{1}$ and $f(v)$ = $v'$ for some $v \in U_{1}$.
\end{description}
\end{definition}

By induction on $\varphi$, we can prove the following. 

\begin{proposition}
\label{prop:bbox-dia-bmor}
Let $M_{i}$ := $(U_{i},H_{i},R_{i},V_{i})$ ($i$ = 1 or 2) be an $H$-model and $f:U_{1} \to U_{2}$ an {\em $\mathcal{L}(\dia,\bbox)$-bounded morphism}. For any formula $\varphi$ in $\mathcal{L}(\dia,\bbox)$ and any state $u \in U_{1}$, the following equivalence holds: 
\[
\text{ $M_{1},u \models \varphi$ $\iff$ $M_{2},f(u) \models \varphi$. }
\]
\end{proposition}

\begin{proposition}\label{prop:undefinable}
\begin{enumerate}
\item The modal operator $\bdia$ is not definable in $\mathcal{L}(\bbox)$. 
\item The modal operator $\Box$ is not definable in $\mathcal{L}(\dia)$. 
\end{enumerate}
\end{proposition}

\begin{proof}
\noindent (1) Suppose for contradiction that $\bdia p$ is definable by a formula $\varphi$ in $\mathcal{L}(\bbox)$. 
Let us define an $H$-model $M_{1}$ = $(U_{1},H_{1},R_{1},V_{1})$ as a pair of an $H$-frame in Example \ref{ex:stable} (1) and a valuation $V_{1}$ defined by $V_{1}(p)$ = $\setof{1,2,3}$, which is an $H$-set. We observe that $M_{1},2 \not\models \bdia p$, i.e., $M_{1},2 \not\models \varphi$ since there is no $R_{1}$-predecessor from $2$.  Let us define an $H$-model $M_{2}$ = $(U_{2},H_{2},R_{2},V_{2})$ as a pair of an $H$-frame in Example \ref{ex:stable} (2) and a valuation $V_{2}$ defined by $V_{2}(p)$ = $\setof{b,c}$, which is an $H$-set. Then we have $M_{2},c \models \bdia p$ hence $M_{2},c \models \varphi$. Consider the mapping $f:U_{1} \to U_{2}$ defined by $f(0)$ = $a$, $f(1)$ = $b$ and $f(2)$ = $f(3)$ = $c$. We can verify that $f$ is an $\mathcal{L}(\dia,\bbox)$-bounded morphism. 
Since $\varphi$ is also a formula in $\mathcal{L}(\dia,\bbox)$, $M_{1},2 \not\models \varphi$ implies $M_{2},f(2) \not\models \varphi$ 
by Proposition \ref{prop:bbox-dia-bmor}. 
But this is a contradiction with $M_{2},c \models \varphi$.

\noindent (2) Suppose for contradiction that $\Box p$ is definable by a formula $\varphi$ in $\mathcal{L}(\bbox)$. Define $H$-models $N_{1}$ and $N_{2}$ as follows (see the first and the third graphs below, where all reflexive $H$-arrows are omitted). 
\begin{itemize}
\item $N_{1}$ = $(U_{1},H_{1},R_{1},V_{1})$ is defined by: $U_{1}$ = $\setof{0,1,2,3}$, $H_{1}$ = $\setof{(1,0),(3,2)} \cup \inset{(u,u)}{u \in U_{1}}$, $R_{1}$ = $\setof{(1,2),(1,3)}$ and $V_{1}(p)$ = $\setof{0,1,2}$, which is an $H$-set. We have $N_{1},0 \models \Box p$ hence $N_{1},0 \models \varphi$ since there is no $R_{1}$-successor. Moreover, $\leftconv R_{1}$ = $\setof{(2,1),(3,1),(2,0),(3,0)}$ (see the dotted double lines in the second graph below). 
\item $N_{2}$ = $(U_{2},H_{2},R_{2},V_{2})$ is defined as follows: $U_{2}$ = $\setof{a,b,c}$, $H_{2}$ = $\setof{(c,b)} \cup \inset{(u,u)}{u \in U_{2}}$, $R_{2}$ = $\setof{(a,b),(a,c)}$, and $V_{2}(p)$ = $\setof{a,b}$, which is an $H$-set. We note that $N_{2},a \not\models \Box p$ hence $N_{1},0 \models \varphi$. Moreover, $\leftconv R_{2}$ = $\setof{(b,a),(c,a)}$ (see the dotted double lines in the fourth graph below).
\end{itemize}
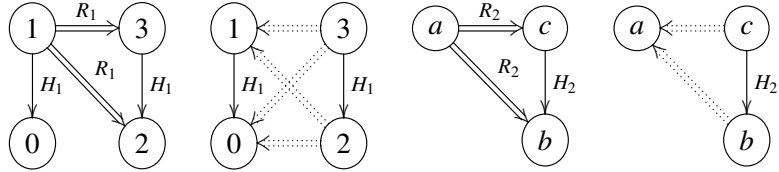
\[
\xymatrix{
*++[o][F-]{1} \ar@{=>}[r]^{R_{1}} \ar@{=>}[dr]^{R_{1}} \ar[d]^{H_{1}} & *++[o][F-]{3} \ar[d]^{H_{1}}\\
*++[o][F-]{0}  & *++[o][F-]{2}  \\
}
\quad
\xymatrix{
*++[o][F-]{1}   \ar[d]^{H_{1}} & *++[o][F-]{3} \ar[d]^{H_{1}} \ar@{:>}[l] \ar@{:>}[ld] \\
*++[o][F-]{0}  & *++[o][F-]{2}  \ar@{:>}[l] \ar@{:>}[lu] \\
}
\quad
\xymatrix{
*++[o][F-]{a} \ar@{=>}[dr]^{R_{2}}\ar@{=>}[r]^{R_{2}}  & *++[o][F-]{c} \ar[d]^{H_{2}} \\
& *++[o][F-]{b} \\
}
\quad
\xymatrix{
*++[o][F-]{a} & *++[o][F-]{c}   \ar[d]^{H_{2}} \ar@{:>}[l]  \\
& *++[o][F-]{b} \ar@{:>}[lu] \\
}
\]
Consider the mapping $f:U_{1} \to U_{2}$ defined by $f(0)$ = $f(1)$ = $a$, $f(2)$ = $b$, $f(3)$ = $c$. Then $f$ is an $\mathcal{L}(\dia,\bbox)$-bounded morphism. By the similar argument as in (1), a contradiction follows. 
\end{proof}


\begin{definition}
We say that a set $\Gamma$ of formula {\em defines} a class $\mathbb{F}$ of $H$-frames if for all $H$-frames $F$, $F \in \mathbb{F}$ iff $F \models \varphi$ for all formulas $\varphi \in \Gamma$. When $\Gamma$ is a singleton $\setof{\varphi}$, we simply say that $\varphi$ defines a class $\mathbb{F}$. 
\end{definition}

The following frame definability results are already established in~\cite[Theorem 10]{Stell2016}. 

\begin{proposition}[\cite{Stell2016}]
\label{prop:definable}
Let $F$ = $(U,H,R)$ be an $H$-frame. Let $S_{i} \in \setof{R,\leftconv R}$ for $i$ = $1, \ldots, m$ and for each $i$ let
\[
\mathsf{B}_{i} =
\begin{cases}
\Box &\text{ if $S_{i}$ = $R$ } \\
\bbox &\text{ if $S_{i}$ = $\leftconv R$ } \\
\end{cases}
\text{ and let }
\mathsf{D}_{i} =
\begin{cases}
\bdia &\text{ if $S_{i}$ = $R$ } \\
\dia &\text{ if $S_{i}$ = $\leftconv R$ } \\
\end{cases}
\]
Let $0 \leqslant k \leqslant m$ (where the composition of a sequence of length 0 is understood as $H$). Then the following are equivalent: (1) $S_{1};\cdots;S_{k} \subseteq S_{k+1};\cdots;S_{m}$; (2) $F \models \mathsf{D}_{k} \cdots \mathsf{D}_{1}p \to \mathsf{D}_{m} \cdots \mathsf{D}_{k+1}p$; (3) $F \models \mathsf{B}_{k+1} \cdots \mathsf{B}_{m} \to \mathsf{B}_{1}p \cdots \mathsf{B}_{k}p$. 
\end{proposition}

Table \ref{table:def} demonstrates the content of Proposition \ref{prop:definable}. 

\begin{table}
\begin{center}
\begin{tabular}{@{}lllll}
\hline
\raisebox{-1.25ex}{\rule{0ex}{4ex}}Condition & Inclusion &
\begin{tabular}{@{}c}Diamond\\Form\end{tabular} &
\begin{tabular}{@{}c}Box\\Form\end{tabular}  &
\begin{tabular}{@{}c}Mixed\\Form\end{tabular} \\
\hline &&&&\\

reflexive             & $H \subseteq R$                        & $ p \to
\bD p$         & $\wB p \to p$         &                         \\[1ex]

\begin{tabular}{@{}l}
converse\\ reflexive
\end{tabular}   & $H \subseteq \leftconv R$                   & $p \to \wD p
$        & $  \bB p \to p$       &                       \\[2ex]

pathetic             & $R \subseteq H$                        & $\bD p
\to p$         & $p \to \wB p$         &
\\[1ex]

\begin{tabular}{@{}l}
converse\\ pathetic
\end{tabular}       & $\leftconv R \subseteq H$                   & $\wD p
\to p$         & $p \to \bB p$         &
\\[2ex]

functional           & $\leftconv R ; R \subseteq H$           & $\bD \wD p
\to p $    & $p \to \bB \wB p$     & $\wD p \to \wB p$   \\[1ex]

injective            & $R ; \leftconv R \subseteq H$           & $\wD \bD p
\to p $    & $p \to \wB \bB p$     & $\bD p \to \bB p$   \\[1ex]

surjective           & $H \subseteq \leftconv R ; R$           & $p \to
\bD \wD p $    & $ \bB \wB p  \to p$   & $\bB p \to \bD p$   \\[1ex]

total                & $H \subseteq R ; \leftconv R$           & $p \to
\wD \bD p $    & $ \wB \bB p  \to p$   & $\wB p \to \wD p$   \\[1ex]

\begin{tabular}{@{}l}
weakly \\ symmetric
\end{tabular}       & $R \subseteq \leftconv R$                   & $ \bD p
\to \wD p $   & $\bB p \to \wB p$     & $p \to \wB \wD p$     \\[2ex]

\begin{tabular}{@{}l}
strongly \\ symmetric
\end{tabular}       & $\leftconv R \subseteq R$                   & $ \wD p
\to \bD p $   & $\wB p \to \bB p$     & $\wD \wB p \to p$     \\[2ex]

transitive      & $R ; R \subseteq R$                & $\bD \bD p
\to \bD p$ & $\wB p \to \wB \wB p$ &          \\[1ex]

\begin{tabular}{@{}l}
converse\\ transitive
\end{tabular} & $\leftconv R ; \leftconv R \subseteq \leftconv R$ & $\wD \wD p \to
\wD p$ & $\bB p \to \bB \bB p$ &          \\[2ex]

dense            & $R  \subseteq R ; R $                & $\bD p \to
\bD \bD p$ & $\wB \wB p \to \wB p$ &          \\[1ex]

\begin{tabular}{@{}l}
converse\\ dense
\end{tabular} & $\leftconv R \subseteq \leftconv R ; \leftconv R$ & $\wD p \to \wD
\wD p$ & $\bB \bB p \to \bB p$ &          \\[2ex]

Euclidean         & $\leftconv R ; R \subseteq R$           & $\bD \wD p
\to \bD p$ & $\wB p \to \bB \wB p$ & $\wD \wB p \to \wB p$ \\[1ex]

\begin{tabular}{@{}l}
weak\\ Euclidean
\end{tabular} & $\leftconv R ; R \subseteq \leftconv R$      & $\bD \wD p \to
\wD p$ & $\bB p \to \bB \wB p$ & $\wD p \to \wB \wD p$ \\[2ex]

\begin{tabular}{@{}l}
converse\\ Euclidean
\end{tabular} & $R ; \leftconv R \subseteq R$ & $\wD \bD p \to \bD p$ &
$\wB p \to \wB \bB p$ &   \\[2ex]

\begin{tabular}{@{}l}
weak converse\\ Euclidean
\end{tabular} & $ R ; \leftconv R \subseteq \leftconv R$ & $\wD \bD p \to \wD
p$ & $\bB p \to \wB \bB p$ &  \\[2ex]

confluent            & $\leftconv R ; R \subseteq R ; \leftconv R$
                                                           & $\bD \wD p
\to \wD \bD p$

          & $\wB \bB p \to \bB \wB p$

                                     & $\wD \wB p \to \wB \wD p$\\[1ex]

divergent         & $R ; \leftconv R \subseteq \leftconv R ; R $
                                                           & $\wD \bD p
\to \bD \wD p$

          & $\bB \wB p \to \wB \bB p$

                                     & $\bD \bB p \to \bB \bD p$\\[1ex]

\hline
\end{tabular}
\end{center}
\caption{\label{CorrTable}Modal correspondents arising from inclusions from~\cite{Stell2016}}
\label{table:def}
\end{table}

\section{Hilbert System of Bi-intuitionistic Stable Tense Logic}
\label{sec:axiomatisation}

In view of the last two items of Proposition \ref{prop:sem-basic}, we employ $\mathcal{L}(\bdia,\Box)$ as our syntax in what follows in this paper and we simply write $\mathcal{L}$ to mean $\mathcal{L}(\bdia,\Box)$ if no confusion arises. 

\begin{definition}
We say that a set $\Lambda$ of formulas is a {\em bi-intuitionistic stable tense logic} (for short, $bist$-logic) if $\Lambda$ contains all the axioms of Table \ref{table:hil-biskt} and closed under all the rules of Table \ref{table:hil-biskt}. Given a $bist$-logic $\Lambda$ and a set $\Gamma \cup \setof{\varphi}$ of formulas, we say that $\varphi$ is {\em $\Lambda$-provable} from $\Gamma$ (notation: $\Gamma \vdash_{\Lambda} \varphi$) if there is a finite set $\Gamma' \subseteq \Gamma$ such that $\bigwedge \Gamma' \to \varphi \in \Lambda$, where $\bigwedge \Delta$ is the conjunction of all elements of $\Delta$ and $\bigwedge \Delta$ := $\top$ when $\Delta$ is an emptyset. Moreover, when $\Gamma$ = $\emptyset$, we simply write $\vdash_{\Lambda} \varphi$ instead of $\emptyset \vdash_{\Lambda} \varphi$, which is equivalent to $\varphi \in \Lambda$. 
We define $\mathbf{BiSKt}$ as the smallest $bist$-logic $\bigcap \inset{\Lambda}{\text{$\Lambda$ is a $bist$-logic}}$. Given a set $\Sigma$ of formulas, the smallest $bist$-logic $\mathbf{BiSKt}\Sigma$ containing $\Sigma$ is defined by: $\mathbf{BiSKt}\Sigma := \bigcap \inset{\Lambda}{\text{$\Lambda$ is a $bist$-logic and $\Sigma \subseteq \Lambda$}}$.
\end{definition}

\noindent Table \ref{table:hil-biskt} provides Hilbert-style axiomatisation of $\mathbf{BiSKt}$. In what follows in this paper, we assume that the reader is familiar with theorems and derived inference rules in intuitionistic logic. 

\begin{table}[htbp]
\caption{Hilbert-style axiomatisation of $\mathsf{H}\mathbf{BiSKt}$ 
}
\label{table:hil-biskt}

\begin{center}
\begin{tabular}{|llll|}
\hline
\multicolumn{4}{|l|}{Axioms and Rules for Intuitionistic Logic}\\
\hline
$(\texttt{A0})$ & \multicolumn{3}{l|}{ $p \to (q \to p)$ } \\
$(\texttt{A1})$ & \multicolumn{3}{l|}{ $(p \to (q \to r)) \to ((p \to q) \to (p \to r))$  } \\
$(\texttt{A2})$ & $p \to (p \lor q)$ & $(\texttt{A3})$ & $q \to (p \lor q)$  \\
$(\texttt{A4})$ & $(p \to r) \to ((q \to r)\to(p \lor q \to r))$ & $(\texttt{A5})$ & $(p \land q) \to p$  \\
$(\texttt{A6})$ & $(p \land q )\to q$  & $(\texttt{A7})$ & $(p \to (q \to p \land q))$  \\
$(\texttt{A8})$ & $\bot \to p$ & $(\texttt{A9})$ & $p \to \top$  \\
$(\texttt{MP})$ & \multicolumn{3}{l|}{ From $\varphi$ and $\varphi \to \psi$, infer $\psi$ } \\
$(\texttt{US})$ & \multicolumn{3}{l|}{ From $\varphi$, infer a substitution instance $\varphi'$ of $\varphi$ }  \\
\hline
\multicolumn{4}{|l|}{Additional Axioms and Rules for Bi-intuitionistic Logic}\\
\hline
$(\texttt{A10})$ & $p \to (q \lor (p \coimp q))$ & $(\texttt{A11})$ & $((q \lor r) \coimp q) \to r$ \\
$(\texttt{Mon}\coimp)$ & \multicolumn{3}{l|}{ From $\delta_{1} \to \delta_{2}$, infer $(\delta_{1} \coimp \psi) \to (\delta_{2} \coimp \psi)$} \\
\hline
\multicolumn{4}{|l|}{Additional Axioms and Rules for Tense Operators}\\
\hline
$(\texttt{A12})$ & $p \to \Box \bdia p$ & $(\texttt{A13})$ & $\bdia \Box p \to p$ \\
$(\texttt{Mon}\Box)$ & From $\varphi \to \psi$, infer $\Box \varphi \to \Box \psi$ & $(\texttt{Mon}\bdia)$ & From $\varphi \to \psi$, infer $\bdia \varphi \to \bdia \psi$ \\
\hline
\end{tabular}
\end{center}
\end{table}

\begin{proposition}
\label{prop:prf-basic}
Let $\Lambda$ be a $bist$-logic. 
\begin{multienumerate}
\mitemxx{$\vdash_{\Lambda} (\varphi \coimp \psi) \to \gamma$ $\iff$ $\vdash_{\Lambda} \varphi \to (\psi \lor \gamma)$.}{$\vdash_{\Lambda} \bdia \varphi \to \psi$ $\iff$ $\vdash_{\Lambda} \varphi \to \Box \psi$. }
\mitemxx{$\vdash_{\Lambda} (\Box \varphi \land \Box \psi) \leftrightarrow \Box (\varphi \land \psi)$}{$\vdash_{\Lambda} \top \leftrightarrow \Box \top$.}
\mitemxx{$\vdash_{\Lambda} (\bdia \varphi \lor \bdia \psi) \leftrightarrow \bdia (\varphi \lor \psi)$}{ $\vdash_{\Lambda} \bot \leftrightarrow \bdia \bot$}
\end{multienumerate}
\end{proposition}

\begin{proof}
Here we prove (1) and (2) alone, since items (3) to (6) are consequence from the adjoint or the residuation  `$\bdia \dashv \Box$' i.e., left adjoints ($\bdia$) preserve colimits (finite disjunctions) and right adjoints ($\Box$) preserve limits (finite conjunctions). 
\begin{enumerate}
\item $(\Rightarrow)$ Assume that $\vdash_{\Lambda} (\varphi \coimp \psi) \to \gamma$. We obtain $\vdash_{\Lambda} (\psi \lor (\varphi \coimp \psi)) \to (\psi \lor \gamma)$. By axiom $(\texttt{A10})$, we get $\vdash_{\Lambda} \varphi \to (\psi \lor \gamma)$.  ($\Leftarrow$) Assume that $\vdash_{\Lambda} \varphi \to (\psi \lor \gamma)$. By rule $(\texttt{Mon}\coimp)$, $\vdash_{\Lambda} (\varphi \coimp \psi) \to ((\psi \lor \gamma) \coimp \psi)$. It follows from $(\texttt{A10})$ that $\vdash_{\Lambda} (\varphi \coimp \psi) \to \gamma$, as desired. 
\item $(\Rightarrow)$ Assume that $\vdash_{\Lambda} \bdia \varphi \to \psi$. By rule $(\texttt{Mon}\Box)$, $\vdash_{\Lambda} \Box  \bdia \varphi \to \Box \psi$. It follows from $(\texttt{A12})$ that $\vdash_{\Lambda} \varphi \to \Box \psi$. ($\Leftarrow$) Suppose that $\vdash_{\Lambda} \varphi \to \Box \psi$. By rule $(\texttt{Mon}\bdia)$, $\vdash_{\Lambda} \bdia \varphi \to \bdia \Box \psi$. By axiom $(\texttt{A13})$, we obtain $\vdash_{\Lambda} \bdia \varphi \to \psi$. 
\end{enumerate}
\end{proof}

\begin{theorem}[Soundness]
\label{thm:sound}
Given any formula $\varphi$, $\vdash_{\mathbf{BiSKt}} \varphi$ implies $\models \varphi$. 
\end{theorem}

\begin{proof}
We demonstrate the soundness of axioms and rules other than those of intuitionistic logic. 
By Proposition \ref{prop:sem-basic}, we have the following equivalences:
\begin{align}
\tag{$\star$}\label{eq:ad-coimp-dis}
\models  (\varphi \coimp \psi) \to \gamma &\iff\, \models \varphi  \to (\psi \lor \gamma)\\
\tag{$\dagger$}\label{eq:ad-bdia-box}
\models \bdia \varphi \to \psi &\iff\, \models \varphi \to \Box \psi 
\end{align}
\begin{itemize}
\item[($\texttt{A10}$)] By $\models (p \coimp q) \to (p \coimp q)$ and (\ref{eq:ad-coimp-dis}), $\models p \to q\lor (p \coimp q)$ holds, as desired. 
\item[($\texttt{A11}$)] By $\models (q \lor r) \to (q \lor r)$ and (\ref{eq:ad-coimp-dis}), we obtain $\models ((q \lor r) \coimp q) \to r$. 
\item[($\texttt{Mon}\coimp$)] Assume $\models \delta_{1} \to \delta_{2}$. To establish $\models (\delta_{1} \coimp \psi) \to (\delta_{2} \coimp \psi)$, it suffices to show $\models \delta_{1}\to (\psi \lor (\delta_{2} \coimp \psi))$. From our assumption and the validity of $(\texttt{A10})$, our goal follows.
\item[($\texttt{A13}$)] By $\models \bdia p \to \bdia p$ and (\ref{eq:ad-bdia-box}), we obtain $\models p \to \Box \bdia p$. 
\item[($\texttt{A14}$)] By $\models \Box p \to \Box p$ and (\ref{eq:ad-bdia-box}), we obtain $\models \bdia \Box p \to p$. 
\item[($\texttt{Mon}\Box$)] Suppose $\models \varphi \to \psi$. To demonstrate $\models \Box \varphi \to \Box \psi$, it suffices to show $\models \bdia \Box \varphi \to \psi$ by (\ref{eq:ad-bdia-box}). By $\models \bdia \Box \varphi \to \varphi$ (due to the validity of $(\texttt{A14})$) and the supposition $\models \varphi \to \psi$), we get our goal. 
\item[($\texttt{Mon}\bdia$)] Suppose $\models \varphi \to \psi$. To conclude $\models \bdia \varphi \to \bdia \psi$, it suffices to show $\models \varphi \to \Box \bdia \psi$ by (\ref{eq:ad-bdia-box}). By the supposition $\models \varphi \to \psi$ and $\models \psi \to \Box \bdia \psi$ (due to the validity of $(\texttt{A13})$), we get our goal. 
\end{itemize}
\end{proof}

\section{Kripke Completeness of Bi-intuitionistic Stable Tense Logic}
\label{sec:completeness}

Given a finite set $\Delta$ of formulas, $\bigvee \Delta$ is defined as the disjunction of all formulas in $\Delta$, where $\bigvee \emptyset$ is understood as $\bot$. 

\begin{definition}
Let $\Lambda$ be a $bist$-logic. A pair $(\Gamma,\Delta)$ of formulas is {\em $\Lambda$-provable} if $\Gamma \vdash_{\Lambda} \bigvee \Delta'$ for some finite $\Delta' \subseteq \Delta$. We say that a pair $(\Gamma,\Delta)$ of formulas is {\em $\Lambda$-unprovable} if it is not $\Lambda$-provable. A pair $(\Gamma,\Delta)$ is {\em complete} if $\Gamma \cup \Delta$ = $\mathsf{Form}_{\mathcal{L}}$, i.e., $\varphi \in \Gamma$ or $\varphi \in \Delta$ for all formulas $\varphi$. 
\end{definition}

We remark that a pair $(\Gamma,\Delta)$ is $\Lambda$-unprovable iff $\not\vdash_{\Lambda} \bigwedge \Gamma' \to \bigvee \Delta'$ for all finite $\Gamma'\subseteq \Gamma$ and all finite $\Delta' \subseteq \Delta$. The following lemma holds because $\Lambda$ `contains' intuitionistic logic. 

\begin{lem}
\label{lem:dist-lat}
Let $\Lambda$ be a $bist$-logic and $(\Gamma,\Delta)$ a complete and $\Lambda$-unprovable pair. 
Then, 
\begin{multienumerate}
\mitemxx{$($$\Gamma \vdash_{\Lambda} \varphi$ implies $\varphi \in \Gamma$$)$ for all formulas $\varphi$}{$\Lambda \subseteq \Gamma$}
\mitemxx{If $\varphi \in \Gamma$ and $\varphi \to \psi \in \Gamma$ then $\psi \in \Gamma$}{$\bot \notin \Gamma$}
\mitemxx{$\top \in \Gamma$}{$\varphi \land \psi \in \Gamma$ iff $\varphi \in \Gamma$ and $\psi \in \Gamma$}
\mitemx{$\varphi \lor \psi \in \Gamma$ iff $\varphi \in \Gamma$ or $\psi \in \Gamma$}
\end{multienumerate}
\end{lem}

\begin{lem}
\label{lem:extension}
Let $\Lambda$ be a $bist$-logic. 
Given a $\Lambda$-unprovable pair $(\Gamma,\Delta)$, there exists a complete and $\Lambda$-unprovable pair $(\Gamma^{+},\Delta^{+})$ such that $\Gamma \subseteq \Gamma^{+}$ and $\Delta \subseteq \Delta^{+}$. 
\end{lem}

\begin{definition}
Let $\Lambda$ be a $bist$-logic. The $\Lambda$-canonical $H$-model $M^{\Lambda}$ = $(U^{\Lambda},H^{\Lambda},R^{\Lambda},V^{\Lambda})$ is defined as follows. 
\begin{itemize}
\item $U^{\Lambda}$ := $\inset{(\Gamma,\Delta)}{\text{ $(\Gamma,\Delta)$ is a complete and $\Lambda$-unprovable pair}}$.
\item $(\Gamma_{1},\Delta_{1})H^{\Lambda}(\Gamma_{2},\Delta_{2})$ iff $\Gamma_{1} \subseteq \Gamma_{2}$. 
\item $(\Gamma_{1},\Delta_{1})R^{\Lambda}(\Gamma_{2},\Delta_{2})$ iff ($\Box \varphi \in \Gamma_{1}$ implies $\varphi \in \Gamma_{2}$)  for all formulas  $\varphi$.
\item $(\Gamma,\Delta)\in V^{\Lambda}(p)$ iff $p \in \Gamma$. 
\end{itemize}
\end{definition}

\noindent It is clear that $H^{\Lambda}$ is not a pre-order but also a partial order. Moreover, $(\Gamma_{1},\Delta_{1})H^{\Lambda}(\Gamma_{2},\Delta_{2})$ implies $\Delta_{2} \subseteq \Delta_{1}$ by completeness. 

\begin{lem}
\label{lem:access}
Let $(\Gamma_{i},\Delta_{i}) \in U^{\Lambda}$ $($$i$ = 1 or 2$)$. The following are all equivalent: 
\begin{enumerate}
\item $(\Gamma_{1},\Delta_{1})R^{\Lambda}(\Gamma_{2},\Delta_{2})$,
\label{item:r1}
\item $($$\varphi \in \Delta_{2}$ implies $\Box \varphi \in \Delta_{1}$$)$ for all formulas  $\varphi$,
\label{item:r2}
 \item  $($$\varphi \in \Gamma_{1}$ implies $\bdia \varphi \in \Gamma_{2}$$)$ for all formulas  $\varphi$,
\label{item:r3}
\item $($$\bdia \varphi \in \Delta_{2}$  implies $\varphi \in \Delta_{1}$$)$ for all formulas $\varphi$.  
\label{item:r4}
\end{enumerate}
\end{lem}

\begin{lem}
\label{lem:stable}
$R^{\Lambda}$ is stable in the $\Lambda$-canonical $H$-model $M^{\Lambda}$. 
\end{lem}

\begin{proof}
It suffices to show show that (i) $H^{\Lambda};R^{\Lambda} \subseteq R^{\Lambda}$ and (ii) $R^{\Lambda};H^{\Lambda}\subseteq R^{\Lambda}$. For (i),  let us assume that $(\Gamma_{1}, \Delta_{1})H^{\Lambda}(\Gamma_{2}, \Delta_{2})R^{\Lambda} (\Gamma_{3}, \Delta_{3})$. To show $(\Gamma_{1}, \Delta_{1})R^{\Lambda}(\Gamma_{3}, \Delta_{3})$, fix any formula $\varphi$ such that $\Box \varphi \in \Gamma_{1}$. Our goal is to show that $\varphi \in \Gamma_{3}$. By $(\Gamma_{1}, \Delta_{1})H^{\Lambda}(\Gamma_{2}, \Delta_{2})$, we get $\Box \varphi \in \Gamma_{2}$. It follows from $(\Gamma_{2}, \Delta_{2})R^{\Lambda} (\Gamma_{3}, \Delta_{3})$ that $\varphi \in \Gamma_{3}$, as required. For (ii), let us assume that $(\Gamma_{1}, \Delta_{1})R^{\Lambda}(\Gamma_{2}, \Delta_{2})H^{\Lambda} (\Gamma_{3}, \Delta_{3})$. To show $(\Gamma_{1}, \Delta_{1})R^{\Lambda}(\Gamma_{3}, \Delta_{3})$, we use Lemma \ref{lem:access} (4) to fix any formula $\bdia \varphi \in \Delta_{3}$. We show that $\varphi \in \Delta_{1}$. Since $\Delta_{3} \subseteq \Delta_{2}$, we have $\bdia \varphi \in \Delta_{2}$. By $(\Gamma_{1}, \Delta_{1})R^{\Lambda}(\Gamma_{2}, \Delta_{2})$, Lemma \ref{lem:access} (4) enables us to conclude $\varphi \in \Delta_{1}$. 
\end{proof}

\begin{lem}
\label{lem:unprov}
Let $(\Gamma,\Delta)$ be a complete $\Lambda$-unprovable pair. 
\begin{enumerate}
\item If $\psi \to \gamma \notin \Gamma$, then $(\setof{\psi}\cup \Gamma,\setof{\gamma})$ is $\Lambda$-unprovable. 
\item If $\psi \coimp \gamma \in \Gamma$, then $(\setof{\psi},\setof{\gamma}\cup \Delta)$ is $\Lambda$-unprovable. 
\item If $\Box \psi \notin \Gamma$, then $(\inset{\gamma}{\Box \gamma \in \Gamma},\setof{\psi})$ is $\Lambda$-unprovable. 
\item If $\bdia \psi \in \Gamma$, then $(\setof{\psi}, \inset{\delta}{\bdia \delta \in \Delta})$ is $\Lambda$-unprovable. 
\end{enumerate}
\end{lem}

\begin{proof}
Here we prove items (2) and (4) alone. For (2), assume that $\psi \coimp \gamma \in \Gamma$. Suppose for contradiction that $(\setof{\psi},\setof{\gamma}\cup \Delta)$ is $\Lambda$-provable. We can find a finite set $\Delta' \subseteq \Delta$ such that $\vdash_{\Lambda} \psi \to (\gamma \lor \bigvee \Delta')$. It follows from Proposition \ref{prop:prf-basic} (1) that $\vdash_{\Lambda} (\psi \coimp \gamma) \to \bigvee \Delta'$. Since $\psi \coimp \gamma \in \Gamma$, $\bigvee \Delta' \in \Gamma$ holds by Lemma \ref{lem:dist-lat}. This contradicts the $\Lambda$-unprovability of $(\Gamma,\Delta)$. For (4), suppose that $\bdia \psi \in \Gamma$. Assume for contradiction that $(\setof{\psi}, \inset{\delta}{\bdia \delta \in \Delta})$ is $\Lambda$-provable. 
There exists some formulas $\delta_{1}$, \ldots, $\delta_{n}$ such that each $\bdia \delta_{i} \in \Delta$ and $\vdash_{\Lambda} \psi \to \bigvee_{1 \leqslant i \leqslant n} \delta_{i}$. By rule $(\texttt{Mon}\bdia)$ (monotonicity of $\bdia$) and commutativity of $\bdia$ over finite disjunctions (due to Proposition \ref{prop:prf-basic} (5) and (6)), it holds that $\vdash_{\Lambda} \bdia \psi \to \bigvee_{1 \leqslant i \leqslant n} \bdia \delta_{i}$. By $\bdia \psi \in \Gamma$ and Lemma \ref{lem:dist-lat}, we obtain $\bigvee_{1 \leqslant i \leqslant n} \bdia \delta_{i} \in \Gamma$, which implies $\bdia \delta_{i} \in \Gamma$ for some indices $i$, again by Lemma \ref{lem:dist-lat}. Fix such index $i$. Together with $\bdia \delta_{i} \in \Delta$, we establish that $(\Gamma,\Delta)$ is $\Lambda$-provable, a contradiction. 
\end{proof}

\begin{lem}[Truth Lemma]
\label{lem:truth}
Let $\Lambda$ be a $bist$-logic. Then, for any formula $\varphi$ and any complete $\Lambda$-unprovable pair $(\Gamma,\Delta)$, the following equivalence holds: $\varphi \in \Gamma$ $\iff$ $M^{\Lambda}, (\Gamma,\Delta) \models \varphi$. 
\end{lem}

\begin{proof}
By induction on $\varphi$. When $\varphi$ is a propositional variable from $\mathsf{Prop}$, it is immediate from the definition of $V^{\Lambda}$. When $\varphi$ is of the form $\psi \land \gamma$ or $\psi \lor \gamma$, we can establish the equivalence by Lemma \ref{lem:dist-lat} and induction hypothesis. In what follows, we deal with the remaining cases, in particular the cases where $\varphi$ is of the form of $\psi \coimp \gamma$ or $\bdia \psi$. 

\noindent \textbf{($\varphi$ is of the form $\psi \coimp \gamma$)}  First we show the right-to-left direction and so assume that $M^{\Lambda},(\Gamma,\Delta) \models \psi \coimp \gamma$. Then there exists a pair $(\Sigma,\Theta) \in U^{\Lambda}$ such that $(\Sigma,\Theta)H^{\Lambda}(\Gamma,\Delta)$ (so $\Sigma \subseteq \Gamma$) and $M^{\Lambda},(\Sigma,\Theta) \models \psi$ but $M^{\Lambda},(\Sigma,\Theta) \not\models \gamma$. By induction hypothesis, $\psi \in \Sigma$ and $\gamma \notin \Sigma$. Our goal is to show that $\psi \coimp \gamma \in \Gamma$. We have $\vdash_{\Lambda} \psi \to (\psi \coimp \gamma) \lor \gamma$ by axiom $(\texttt{A10})$. Since $\psi \in \Sigma$, we obtain $(\psi \coimp \gamma) \lor \gamma \in \Sigma$ hence $\psi \coimp \gamma\in \Sigma$ or  $\gamma \in \Sigma$ by Lemma \ref{lem:dist-lat}. 
By $\gamma \notin \Sigma$, $\psi \coimp \gamma\in \Sigma$ holds. It follows from $(\Sigma,\Theta)H^{\Lambda}(\Gamma,\Delta)$ that $\psi \coimp \gamma\in \Sigma$, as required. Second we show the left-to-right direction. Suppose that $\psi \coimp \gamma \in \Gamma$. By Lemma \ref{lem:unprov} (2), $(\setof{\psi},\setof{\gamma}\cup \Delta)$ is $\Lambda$-unprovable. By Lemma \ref{lem:extension}, we can find $(\Sigma,\Theta) \in U^{\Lambda}$ such that $\psi \in \Sigma$ and $\setof{\gamma}\cup \Delta \subseteq \Theta$. 
It follows from $\Delta \subseteq \Theta$ that $(\Sigma,\Theta)H^{\Lambda}(\Gamma,\Delta)$. By induction hypothesis, we also obtain $M^{\Lambda},(\Sigma,\Theta)\models\psi$ but  $M^{\Lambda},(\Sigma,\Theta)\not\models\gamma$. Therefore, $M^{\Lambda},(\Gamma,\Delta) \models \psi \coimp \gamma$, as desired.

\noindent \textbf{($\varphi$ is of the form $\bdia \psi$)} First we show the right-to-left direction and so assume that $M^{\Lambda},(\Gamma,\Delta) \models \bdia \psi$. We can find a pair $(\Sigma,\Theta) \in U^{\Lambda}$ such that $(\Sigma,\Theta) R^{\Lambda} (\Gamma,\Delta)$ and $M^{\Lambda}, (\Sigma,\Theta) \models \psi$. By induction hypothesis, we get $\psi \in \Sigma$. By Lemma \ref{lem:access}, $\bdia \psi \in \Gamma$, as desired. Second we show the left-to-right direction. Suppose that $\bdia \psi \in \Gamma$. By Lemma \ref{lem:unprov}, $(\setof{\psi}, \inset{\delta}{\bdia \delta \in \Delta})$ is $\Lambda$-unprovable. By Lemma \ref{lem:extension}, we can find a pair $(\Sigma,\Theta) \in U^{\Lambda}$ such that $\psi \in \Sigma$ and $\inset{\delta}{\bdia \delta \in \Delta} \subseteq \Theta$. By Lemma \ref{lem:access} and induction hypothesis, we obtain $(\Sigma,\Theta)R^{\Sigma}(\Gamma,\Delta)$ and $\psi \in \Sigma$. Therefore, $M^{\Lambda}, (\Gamma,\Delta)\models \bdia\psi$. 
\end{proof}

\begin{theorem}[Strong Completeness of $\mathbf{BiSKt}$]
\label{thm:complete}
Given any set $\Gamma \cup \setof{\varphi}$ of formulas, 
\[
\text{ $\Gamma \models \varphi$ implies $\Gamma \vdash_{\mathbf{BiSKt}} \varphi$. }
\]
\end{theorem}
\begin{proof}
Put $\Lambda$ := $\mathbf{BiSKt}$. Fix any set $\Gamma \cup \setof{\varphi}$ of formulas. We prove the contrapositive implication and so assume that $\Gamma \not\vdash_{\Lambda} \varphi$. It follows that $(\Gamma,\setof{\varphi})$ is $\Lambda$-unprovable. By Lemma \ref{lem:extension},  we can find a complete and $\Lambda$-unprovable pair $(\Sigma,\Theta) \in U^{\Lambda}$ such that $\Gamma \subseteq \Sigma$ and $\varphi \in \Theta$. By Lemma \ref{lem:truth} (Truth Lemma), $M^{\Lambda},(\Sigma,\Theta) \models \gamma$ for all $\gamma \in \Gamma$ and $M^{\Lambda},(\Sigma,\Theta) \not\models \varphi$. Since $M^{\Lambda}$ is an $H$-model by Lemma \ref{lem:stable}, we can conclude $\Gamma \not\models \varphi$, as desired. 
\end{proof}

\section{Kripke Completeness of Extensions of BiSKt}
\label{sec:extension}

This section establishes that a $bist$-logic extended with {\em any} set of formulas from Table \ref{table:def} enjoys the strongly completeness.

\begin{definition}
Let $\mathbb{F}$ be a frame class. We say that $\varphi$ is a {\em semantic $\mathbb{F}$-consequence from} $\Gamma$ (written: $\Gamma \models_{\mathbb{F}} \varphi$) if, whenever $(F,V),u \models \gamma$ for all $\gamma \in \Gamma$, it holds that $(F,V),u\models \varphi$, for any $H$-frame $F$ = $(U,H,R) \in \mathbb{F}$, any valuation $V$ on $U$ and any state $u \in U$. 
\end{definition}

\noindent When $\Gamma$ is empty in the notation `$\Gamma \models_{\mathbb{F}} \varphi$', we simply write $\models_{\mathbb{F}} \varphi$, which is equivalent to the statement that $F \models \varphi$ for all $F \in \mathbb{F}$.  

Given a set $\Sigma$ of formulas, recall that $\mathbf{BiSKt}\Sigma$ is the smallest $bist$-logic containing $\Sigma$. By Proposition \ref{prop:definable}, we obtain the following soundness result. 

\begin{theorem}
\label{thm:sound-extension}
Let $\Sigma$ be a possibly infinite set of formulas of the form $\mathsf{D}_{k} \cdots \mathsf{D}_{1}p \to \mathsf{D}_{m} \cdots \mathsf{D}_{k+1}p$ and $\mathbb{F}_{\Sigma}$ be the class of $H$-frames defined by $\Sigma$. 
Then $\mathbf{BiSKt}\Sigma$ is sound for the class $\mathbb{F}_{\Sigma}$, i.e., $\vdash_{\mathbf{BiSKt}\Sigma }\varphi$ implies $\mathbb{F}_{\Sigma} \models \varphi$, for all formulas $\varphi$. 
\end{theorem}

In what follows, we show that, for any set $\Sigma$ of formulas of the form $\mathsf{D}_{k} \cdots \mathsf{D}_{1}p \to \mathsf{D}_{m} \cdots \mathsf{D}_{k+1}p$, $\mathbf{BiSKt}\Sigma$ is strongly complete for the class of $H$-frames defined by $\Sigma$. 

\begin{proposition}
\label{prop:coneg-dia}
Let $\Lambda$ be a $bist$-logic. 
\begin{multienumerate}
\mitemxx{$\vdash_{\Lambda} \varphi \to \neg \psi$ iff $\vdash_{\Lambda} \psi \to \neg \varphi$.}{$\vdash_{\Lambda} \coneg \varphi \to \psi$ iff $\vdash_{\Lambda} \coneg \psi \to \varphi$. }
\mitemxx{$\vdash_{\Lambda} \dia \varphi \to \psi$ iff $\vdash_{\Lambda}  \varphi \to  \bbox \psi$.}{$\vdash_{\Lambda}  \dia \bot \leftrightarrow \bot$.}
\mitemx{$\vdash_{\Lambda}  \dia(\varphi \lor \psi) \leftrightarrow (\dia \varphi \lor \dia \psi)$. }
\end{multienumerate}
\end{proposition}

\begin{proof}
(4) and (5) follows from (3) similarly as in the proof of Proposition \ref{prop:prf-basic}, i.e., `left adjoints ($\dia$) preserves colimits (finite disjunctions).' Recall that $\dia$ := $\coneg \Box \neg$ and $\bbox$ := $\neg \bdia \coneg$. 
(1) is easy to show. So, we focus on items  (2) and (3). 
For (2), let us first recall that $\coneg \varphi$ := $\top \coimp \varphi$. 
By Proposition \ref{prop:prf-basic}, we proceed as follows: $\vdash_{\Lambda} \coneg \varphi \to \psi$ iff $\vdash_{\Lambda} \top \to (\varphi \lor \psi)$ iff $\vdash_{\Lambda} \top \to (\psi \lor \varphi)$ iff  $\vdash_{\Lambda} \coneg \psi \to \varphi$. We finished to prove (2). For (3), we proceed as follows: $\vdash_{\Lambda} \coneg \Box \neg \varphi \to \psi$ iff 
$\vdash_{\Lambda} \coneg \psi \to \Box \neg \varphi$ (by item (2)) iff 
$\vdash_{\Lambda} \bdia \coneg \psi \to \neg \varphi$ (by Proposition \ref{prop:prf-basic}) 
iff $\vdash_{\Lambda} \varphi \to \neg \bdia \coneg \psi$ (by item (1)). 
\end{proof}

\begin{lem}
\label{lem:leftconv}
Let $\Lambda$ be a $bist$-logic and $(\Gamma,\Delta), (\Sigma,\Theta) \in U^{\Lambda}$. Then,
\[
\begin{array}{lll}
(\Gamma,\Delta) H^{\Lambda};\breve{R^{\Lambda}};H^{\Lambda}(\Sigma,\Theta) &\iff& \inset{\varphi }{\coneg \Box \neg \varphi \in \Theta} \subseteq \Delta.\\
\end{array}
\]
Therefore, $(\Gamma,\Delta) \leftconv R^{\Lambda} (\Sigma,\Theta)$ iff $($$\dia \varphi \in \Theta$ implies $\varphi \in \Delta$$)$ for all formulas $\varphi$. 
\end{lem}

\begin{lem}
\label{lem:definable}
Let $\Lambda$ be a $bist$-logic and $S_{i}^{\Lambda} \in \setof{R^{\Lambda},\leftconv R^{\Lambda}}$ for $1 \leqslant i \leqslant m$ and for each $i$ let $\mathsf{D}_{i}$ be $\bdia$ if $S_{i}^{\Lambda}$ = $R^{\Lambda}$; $\dia$ if $S_{i}^{\Lambda}$ = $\leftconv R^{\Lambda}$. For all pairs $(\Gamma,\Delta),(\Sigma,\Theta)\in U^{\Lambda}$, we have the following equivalence: $(\Gamma,\Delta) S_{1}^{\Lambda}; \cdots ;S_{m}^{\Lambda}(\Sigma,\Theta)$ $\iff$ $\inset{\varphi}{\mathsf{D}_{m} \cdots \mathsf{D}_{1} \varphi \in \Theta } \subseteq \Delta$. 
\end{lem}

\begin{proof}
By induction on $m$. (\textbf{Basis}) When $m$ = 0, we need to show the equivalence: $(\Gamma,\Delta) H^{\Lambda}(\Sigma,\Theta)$ iff $\Theta \subseteq \Delta$. This is immediate. (\textbf{Inductive Step}) Let $m$ = $k+1$. We show the equivalence:
\[
\begin{array}{lll}
(\Gamma,\Delta) S_{1}^{\Lambda}; \cdots ;S_{k+1}^{\Lambda}(\Sigma,\Theta) &\iff& \inset{\varphi}{\mathsf{D}_{k+1} \cdots \mathsf{D}_{1} \varphi \in \Theta } \subseteq \Delta. 
\end{array}
\]
By Lemmas \ref{lem:access} and \ref{lem:leftconv}, the left-to-right direction is easy to establish, so we focus on the converse direction. Assume $\inset{\varphi}{\mathsf{D}_{k+1} \cdots \mathsf{D}_{1} \varphi \in \Theta } \subseteq \Delta$. We show that there exists a pair $(\Gamma_{1},\Delta_{1}) \in U^{\Lambda}$ such that $(\Gamma,\Delta) S_{1}^{\Lambda}(\Gamma_{1},\Delta_{1})$ and $(\Gamma_{1},\Delta_{1})S_{2}^{\Lambda}; \cdots ;S_{k+1}^{\Lambda}(\Sigma,\Theta)$. It suffices to show that 
\[
(\inset{\mathsf{D}_{1} \gamma}{\gamma \in \Gamma},\inset{\varphi}{\mathsf{D}_{k+1} \cdots \mathsf{D}_{2} \varphi \in \Theta })
\] 
is $\Lambda$-unprovable. This because Lemmas \ref{lem:access} and \ref{lem:leftconv} and our induction hypothesis allow us to get the desired goal. Suppose otherwise. Then we can find a formulas $\gamma_{1},\ldots,\gamma_{a} \in \Gamma$ and $\varphi_{1}, \ldots, \varphi_{b}$ such that $\mathsf{D}_{k+1} \cdots \mathsf{D}_{2} \varphi_{j} \in \Theta$ for all indices $j$ and $\vdash_{\Lambda} {\bigwedge}_{i \in I} \mathsf{D}_{1} \gamma_{i} \to {\bigvee}_{j \in J} \varphi_{j}$ where $I$ = $\setof{1,\dots,a}$ and $J$ = $\setof{1,\dots,b}$. 
Now by monotonicity of $\mathsf{D}_{1}$ (due to  Propositions \ref{prop:prf-basic} and \ref{prop:coneg-dia}), $\vdash_{\Lambda}  \mathsf{D}_{1}{\bigwedge}_{i \in I} \gamma_{i} \to {\bigvee}_{j \in J} \varphi_{j}$. 
By monotonicity of `$\mathsf{D}_{k+1} \cdots \mathsf{D}_{2}$' (by Propositions \ref{prop:prf-basic} and \ref{prop:coneg-dia}), $\vdash_{\Lambda}  \mathsf{D}_{k+1} \cdots \mathsf{D}_{2}  \mathsf{D}_{1}{\bigwedge}_{i \in I} \gamma_{i} \to \mathsf{D}_{k+1} \cdots \mathsf{D}_{2} {\bigvee}_{j \in J} \varphi_{j}$. By commutativity of `$\mathsf{D}_{k+1} \cdots \mathsf{D}_{2}$' over finite disjunctions (again due to  Propositions \ref{prop:prf-basic} and \ref{prop:coneg-dia}), 
\[
\vdash_{\Lambda}  \mathsf{D}_{k+1} \cdots \mathsf{D}_{2}  \mathsf{D}_{1}{\bigwedge}_{i \in I} \gamma_{i} \to {\bigvee}_{j \in J}  \mathsf{D}_{k+1} \cdots \mathsf{D}_{2} \varphi_{j}. 
\]
Since $\mathsf{D}_{k+1} \cdots \mathsf{D}_{2} \varphi_{j} \in \Theta$ for all $j \in J$, we get ${\bigvee}_{j \in J}  \mathsf{D}_{k+1} \cdots \mathsf{D}_{2} \varphi_{j} \in \Theta$. By the implication established above, we obtain $\mathsf{D}_{k+1} \cdots \mathsf{D}_{2}  \mathsf{D}_{1}{\bigwedge}_{i \in I} \gamma_{i} \in \Theta$. By our initial assumption of $\inset{\varphi}{\mathsf{D}_{k+1} \cdots \mathsf{D}_{1} \varphi \in \Theta } \subseteq \Delta$, we obtain ${\bigwedge}_{i \in I} \gamma_{i}  \in \Delta$. On the other hand, since all $\delta_{i}$s belong to $\Gamma$, this implies that ${\bigwedge}_{i \in I} \gamma_{i} \in \Gamma$, a contradiction with $\Lambda$-unprovability of $(\Gamma,\Delta)$. 
\end{proof}

\begin{theorem}
\label{thm:complete-extension}
Let $\Sigma$ be a possibly infinite set of formulas of the form $\mathsf{D}_{k} \cdots \mathsf{D}_{1}p \to \mathsf{D}_{m} \cdots \mathsf{D}_{k+1}p$ $($where $\mathsf{D}_{i} \in \setof{\bdia,\dia}$$)$ and $\mathbb{F}_{\Sigma}$ be the class of $H$-frames defined by $\Sigma$. Then $\mathbf{BiSKt}\Sigma$ is strongly complete for the class $\mathbb{F}_{\Sigma}$, i.e., if $\Gamma \models_{\mathbb{F}_{\Sigma}} \varphi$ then $\Gamma \vdash_{\mathbf{BiSKt}\Sigma }\varphi$, for all sets $\Gamma \cup \setof{\varphi}$ of formulas. 
\end{theorem}

\begin{proof}
Let us put $\Lambda$ := $\mathbf{BiSKt}\Sigma$. Suppose that $\Gamma \not\vdash_{\Lambda}\varphi$. Our argument for $\Gamma \models_{\mathbb{F}_{\Sigma}} \varphi$ is almost the same as in the proof of Theorem \ref{thm:complete}, but we need to check that the frame part $F^{\Lambda}$ = $(U^{\Lambda},H^{\Lambda},R^{\Lambda})$ of the $\Lambda$-canonical model $M^{\Lambda}$ belongs to the class $\mathbb{F}_{\Sigma}$.  In the notation of Lemma \ref{lem:definable}, it suffices to show that $S_{1}^{\Lambda};\cdots;S_{k}^{\Lambda} \subseteq S_{k+1}^{\Lambda};\cdots;S_{m}^{\Lambda}$ for any formula $\mathsf{D}_{k} \cdots \mathsf{D}_{1}p \to \mathsf{D}_{m} \cdots \mathsf{D}_{k+1}p$ from $\Sigma$.  Suppose that $(\Gamma,\Delta)S_{1}^{\Lambda};\cdots;S_{k}^{\Lambda} (\Gamma',\Delta')$. To show 
$(\Gamma,\Delta) S_{k+1}^{\Lambda};\cdots;S_{m}^{\Lambda} (\Gamma',\Delta')$, we assume that $\mathsf{D}_{m}\cdots\mathsf{D}_{k+1} \varphi \in \Delta'$ by Lemma \ref{lem:definable}. Our goal is to establish $\varphi \in \Delta$. Since $\mathsf{D}_{k} \cdots \mathsf{D}_{1}p \to \mathsf{D}_{m} \cdots \mathsf{D}_{k+1}p \in \Lambda$, it holds that $\vdash_{\Lambda} \mathsf{D}_{k} \cdots \mathsf{D}_{1}\varphi \to \mathsf{D}_{m} \cdots \mathsf{D}_{k+1}\varphi$. It follows from our assumption that $\mathsf{D}_{k} \cdots \mathsf{D}_{1} \varphi \in \Delta'$. By the supposition $(\Gamma,\Delta)S_{1}^{\Lambda};\cdots;S_{k}^{\Lambda} (\Gamma',\Delta')$, Lemma \ref{lem:definable} allows us to conclude $\varphi \in \Delta$. Therefore, $F^{\Lambda} \in \mathbb{F}_{\Sigma}$, as required. 
\end{proof}

\begin{corollary}
\label{cor:strong-comp-table}
Let $\Sigma$ be a set of formulas from Table \ref{table:def} and $\mathbb{F}_{\Sigma}$ be the class of $H$-frames defined by $\Sigma$. Then $\mathbf{BiSKt}\Sigma$ is strongly complete for the class $\mathbb{F}_{\Sigma}$.
\end{corollary}

\section{Finite Model Property for Bi-intuitionistic Stable Tense Logics}
\label{sec:fmp}

A $bist$-logic $\Lambda$ has the {\em finite model property} if for every non-theorem $\varphi \notin \Lambda$, there is a finite frame $F$ such that $F \models \Lambda$ but $F \not\models \varphi$. 
We say that a $bist$-logic $\Lambda$ is {\em finitely axiomatizable} if $\Lambda$ = $\mathbf{BiSKt}\Sigma$ for some finite set $\Sigma$ of formulas. It is well-known that if $\Lambda$ is finitely axiomatizable and has the finite model property then it is decidable. In this section, we show that some strongly complete extensions of $\mathbf{BiSKt}$ enjoy the finite model property and so the decidability. We employ the filtration method in~\cite{Hasimoto2001} for intuitionistic modal logics to establish the finite model property also for some $bist$-logics. 

Let $M$ = $(U,H,R,V)$ be an $H$-model and $\Delta$ a subformula closed set of formulas. We define an equivalence relation $\sim_{\Delta}$ by: $x \sim_{\Delta} y$ $\iff$ $(M,x \models \varphi \text{ iff } M,y \models \varphi)$ for all $\varphi \in \Sigma$. 
When $x \sim_{\Delta} y$ holds, we say that $x$ and $y$ are $\Delta$-equivalent. We use $[x]$ to mean an equivalence class $\inset{y \in U}{x \sim_{\Delta} y}$ of $x \in U$. 

\begin{definition}[Filtration]
We say that a model $M_{\Delta}$ = $(U_{\Delta}, H_{\Delta}, R_{\Delta}, V_{\Delta})$ is a {\em filtration} of an $H$-model $M$ = $(U,H,R,V)$ through a subformula closed set $\Delta$ of formulas if the following conditions are satisfied. 
\begin{enumerate}
\item $U_{\Delta}$ = $\inset{[x]}{x \in U}$. 
\item For all $x,y \in U$, if $xHy$ then $[x]H_{\Delta}[y]$.
\item For all $x,y \in U$ and $\varphi \in \Delta$, if $[x]H_{\Delta}[y]$ and $M,x \models \varphi$ then $M,y \models \varphi$. 
\item For all $x,y \in U$, if $xRy$ then $[x]R_{\Delta}[y]$.
\item For all $x,y \in U$ and $\Box \varphi \in \Delta$, if $[x]R_{\Delta}[y]$ and $M,x \models \Box \varphi$ then $M,y \models \varphi$. 
\item For all $x,y \in U$ and $\bdia \varphi \in \Delta$, if $[x]R_{\Delta}[y]$ and $M,x \models \varphi$ then $M,y \models \bdia \varphi$. 
\item $V_{\Delta}(p)$ = $\inset{[x]}{x \in V(p)}$ for all $p \in \Delta$.
\end{enumerate}
\end{definition}

When $\Delta$ is finite, we note that $U_{\Delta}$ is also finite. 

\begin{proposition}
\label{prop:fil-lem}
Let $M_{\Delta}$ = $(U_{\Delta}, H_{\Delta}, R_{\Delta}, V_{\Delta})$ is a filtration of an $H$- model $M$ = $(U,H,R,V)$ through a subformula closed set $\Delta$ of formulas. Then for every $x \in U$ and every $\varphi \in \Delta$, the following equivalence holds: $M,x \models \varphi$ $\iff$ $M_{\Delta},[x] \models \varphi$. 
\end{proposition}

\begin{proof}
We only show the case where $\varphi$ is of the form $\psi \coimp \gamma$. Note that $\psi,\gamma \in \Delta$. 
For the left-to-right direction, assume $M,x \models \psi \coimp \gamma$, i.e., there exists $y \in U$ such that $yHx$ and $M,y \models \psi$ and $M,y \not\models \gamma$. 
Fix such $y$. By the condition (2) of filtration, $[y]H_{\Delta}[x]$. 
It follows from induction hypothesis that $M_{\Delta},[y] \models \psi$ and $M_{\Delta},[y] \not\models \gamma$. 
Therefore, $M_{\Delta},[x]\models \psi \coimp \gamma$. 
For the right-to-left direction, assume $M_{\Delta},[x] \models \psi \coimp \gamma$. 
So, we can find an equivalence class $[y]$ such that $[y]H_{\Delta}[x]$ and 
$M_{\Delta},[y] \models \psi$ and $M_{\Delta},[y] \not\models \gamma$. 
By induction hypothesis, $M,y \models \psi$ and $M,y \not\models \gamma$. 
Since $yHy$, we obtain $M,y \models \psi \coimp \gamma$. 
By $[y]H_{\Delta}[x]$ and the condition (3) of filtration, $M,x \models \psi \coimp \gamma$ holds.
\end{proof}

While our definition of filtration does not guarantee us the {\em existence} of a filtration, the following definition and proposition provide an example of filtration. We remark that the filtration in Definition \ref{dfn:fin-fil} is called the {\em finest filtration} and is shown to be the smallest filtration in~\cite{Hasimoto2001}.

\begin{definition}
\label{dfn:fin-fil}
Given an $H$-frame $F$ = $(U,H,R)$ and a subformula closed set $\Delta$, $\underline{H}_{\Delta}$ and $\underline{R}_{\Delta}$ are defined by:
\[
\begin{array}{lll}
[x] \underline{H}_{\Delta} [y] & \iff & x'Hy' \text{ for some $x' \in [x]$ and some $y' \in [y]$, } \\

[x] \underline{R}_{\Delta} [y] & \iff & x'Ry' \text{ for some $x' \in [x]$ and some $y' \in [y]$. } \\
\end{array}
\]
Put $\underline{R}_{\Delta}^{s}$ := $\underline{H}_{\Delta}^{+};\underline{R}_{\Delta};\underline{H}_{\Delta}^{+}$ where $X^{+}$ is the transitive closure of a binary relation $X$ on a set. 
\end{definition}

As noted in~\cite{Hasimoto2001}, we need to take the `stability-closure' of $\underline{R}_{\Delta}$ to define $\underline{R}_{\Delta}^{s}$, because we cannot always assure that $\underline{R}_{\Delta}$ is stable. 

\begin{proposition}
\label{prop::fin-fil}
Let $M$ = $(U,H,R,V)$ be an $H$-model and $\Delta$ a subformula closed set. Then $M_{\Delta}^{s}$ := $(U_{\Delta},\underline{H}^{+}_{\Delta},\underline{R}_{\Delta}^{s},V_{\Delta})$ is an $H$-model and a filtration of $M$ through $\Delta$. 
\end{proposition}

\begin{proof}
We show that all conditions for filtration are satified. We check the condition (6) alone, since the others are already shown in~\cite{Hasimoto2001} for intuitionistic modal logics. \\
\noindent (6) Assume $[x]\underline{R}^{s}_{\Delta}[y]$ and $M,x \models \varphi$. We show that $M,y \models \bdia \varphi$. 
By $[x]\underline{R}^{s}_{\Delta}[y]$, there exists $[x'], [y'] \in U_{\Delta}$ such that $[x]\underline{H}_{\Delta}^{+}[x']$ and $[x']\underline{R}_{\Delta}[y']$ and $[y']\underline{H}_{\Delta}^{+}[y]$. It follows from (3) in the above, our assumption and $[x]\underline{H}_{\Delta}^{+}[x']$ that $M,x' \models \varphi$. By $[x']\underline{R}_{\Delta}[y']$, there exists $x'' \in [x']$ and $y'' \in [y']$ such that $x''Ry''$. We obtain $M,x'' \models \varphi$ by $x'' \in [x']$. So $M,y'' \models \bdia \varphi$. By $y'' \in [y']$, $M,y' \models \bdia \varphi$. By (3) in the above and $[y']\underline{H}_{\Delta}^{+}[y]$, we can conclude $M,y \models \bdia \varphi$.
\end{proof}

\begin{proposition}
\label{prop:preserve-finitest}
Let $F$ = $(U,H,R)$ be an $H$-frame. Let $S_{i} \in \setof{R,\leftconv R}$ for $1 \leqslant i \leqslant m$. 
\begin{enumerate}
\item If $(x,y) \in \leftconv {R}$ then $([x],[y]) \in \leftconv \underline{R}_{\Delta}^{s}$. 
\item If $F$ satisfies $H \subseteq S_{1};\cdots;S_{m}$ then $(U_{\Sigma},\underline{H}^{+}_{\Delta},\underline{R}_{\Delta}^{s})$ also satisfies the corresponding property. 
\item If $F$ satisfies $R \subseteq S_{1};\cdots;S_{m}$ then $(U_{\Sigma},\underline{H}^{+}_{\Delta},\underline{R}_{\Delta}^{s})$ also satisfies the corresponding property. 
\end{enumerate}
\end{proposition}

\begin{proof}
We focus on items (1) and (3) alone. For (1), assume that $(x,y) \in \leftconv R$. This means that $xHx'$, $y'Rx'$ and $y'Hy$ for some $x',y' \in U$. 
It follows that $[y']\underline{H}_{\Delta}[y']\underline{R}_{\Delta}[x']\underline{H}_{\Delta}[x']$ hence $[y']\underline{R}_{\Delta}^{s}[x']$. By assumption, we also obtain $[x]\underline{H}_{\Delta}^{+}[x']$ and $[y']\underline{H}_{\Delta}^{+}[y]$. 
We can now conclude $([x],[y]) \in \leftconv \underline{R}_{\Delta}^{s}$. 
For (3), assume that $R \subseteq S_{1};\cdots;S_{m}$. We show that $\underline{R}_{\Delta}^{s} \subseteq (S_{1})_{\Delta}; \cdots ;(S_{m})_{\Delta}$ and so suppose that $[x] \underline{R}_{\Delta}^{s} [y]$. 
This means that $[x] \underline{H}_{\Delta}^{+} [x'] \underline{R}_{\Delta} [y'] \underline{H}_{\Delta}^{+} [x]$ for some $[x']$, $[y'] \in U_{\Delta}$. By definition, there exist $a \in [x']$ and $b \in [y']$ such that $a R b$. 
By assumption, $(a,b)\in S_{1};\cdots;S_{m}$. By item (1) and the condition (4) of filtration for $\underline{R}_{\Delta}^{s}$, we obtain $([x'],[y'])\in (S_{1})_{\Delta};\cdots;(S_{m})_{\Delta}$. Since $(S_{i})_{\Delta}$ is stable, it follows from $[x] \underline{H}_{\Delta}^{+} [x']$ and $[y'] \underline{H}_{\Delta}^{+} [x]$ that $([x],[y])\in (S_{1})_{\Delta};\cdots;(S_{m})_{\Delta}$, as desired. 
\end{proof}

\begin{theorem}
\label{thm:fmp-fin}
Let $\Sigma$ be a possibly empty {\em finite} set of formulas of the form $p \to \mathsf{D}_{1} \cdots \mathsf{D}_{m}p$ or $\bdia p \to \mathsf{D}_{1} \cdots \mathsf{D}_{m}p$ $($where $\mathsf{D}_{i} \in \setof{\bdia,\dia}$$)$. Then $\mathbf{BiSKt}\Sigma$ enjoys the finite model property. Therefore, $\mathbf{BiSKt}\Sigma$ is decidable. 
\end{theorem}

\begin{proof}
It suffices to show the former part, i.e., the finite model property. 
Let $\varphi \notin \mathbf{BiSKt}\Sigma$, i.e., $\varphi$ is a non-theorem of $\mathbf{BiSKt}\Sigma$. 
By Theorem \ref{thm:complete-extension}, there is a model $M$ = $(U,H,R,V)$ such that $M \not\models \varphi$ and $F$ = $(U,R,V) \models \Sigma$. By Proposition \ref{prop:definable}, $F$ satisfies the corresponding properties to all the elements of $\Sigma$. 
Put $\Delta$ as the set of all subformulas of $\varphi$. By Proposition \ref{prop:fil-lem}, we obtain $M_{\Delta}^{s} \not\models \varphi$ hence $F_{\Delta}^{s} \not\models \varphi$ where $F_{\Delta}^{s}$ is the frame part of $M_{\Delta}^{s}$. Moreover, Proposition \ref{prop:preserve-finitest} implies $F_{\Delta}^{s} \models \Sigma$, which implies $F_{\Delta}^{s} \models \mathbf{BiSKt}\Sigma$. 
\end{proof}

When $\Sigma$ is a set of formulas from Table \ref{table:def} which satisfies the syntactic condition in the statement of Theorem \ref{thm:fmp-fin}, then the theorem implies that $\mathbf{BiSKt}\Sigma$ is always decidable. In particular:

\begin{corollary}[\cite{Stell2016}]
\label{cor:dec-biskt}
$\mathbf{BiSKt}$ is decidable. 
\end{corollary}

When $\Sigma$ contains the formula $\bdia \bdia p \to \bdia p$ which defines the transitivity of $R$ (recall Table \ref{table:def}), we have the following partial results on the decidability of $\mathbf{BiSKt}\Sigma$. 

\begin{proposition}
\label{prop:tra-fil}
Let $M$ = $(U,H,R,V)$ be an $H$-model and $\Delta$ a subformula closed set. If $R$ is transitive, then $M_{\Delta}^{s+}$ := $(U_{\Delta},\underline{H}^{+}_{\Delta},(\underline{R}_{\Delta}^{s})^{+},V_{\Delta})$ is an $H$-model and a filtration of $M$ through $\Delta$, where $(\underline{R}_{\Delta}^{s})^{+}$ is the transitive closure of $\underline{R}_{\Delta}^{s}$. 
\end{proposition}

\begin{theorem}
\label{thm:fmp-K4-S4}
$\mathbf{BiSKt}\setof{\bdia \bdia p \to \bdia p}$ and 
$\mathbf{BiSKt}\setof{\bdia \bdia p \to \bdia p, p \to \bdia p}$ enjoy the finite model property.
Therefore, they are decidable.  
\end{theorem}

\begin{proof}
Put $\Lambda$ as one of the $bist$-logics in the statement. We only show the finite model property. 
Let $\varphi \notin \Lambda$, i.e., $\varphi$ is a non-theorem of $\Lambda$. 
By Theorem \ref{thm:complete-extension}, there is an $H$-model $M$ = $(U,H,R,V)$ such that $M \not\models \varphi$ and $F$ = $(U,R,V) \models \Sigma$. By Proposition \ref{prop:definable}, $F$ satisfies the corresponding properties to $\bdia \bdia p \to \bdia p$, $p \to \bdia p$ in $\Lambda$. Put $\Delta$ is the set of all subformulas of $\varphi$. By Propositions \ref{prop:fil-lem} and \ref{prop:tra-fil}, we obtain $M_{\Delta}^{s+} \not\models \varphi$ hence $F_{\Delta}^{s+} \not\models \varphi$ where $F_{\Delta}^{s+}$ is the frame part of $M_{\Delta}^{s+}$. We note that $(\underline{R}_{\Delta}^{s})^{+}$ is clearly transitive. If $R$ is reflexive, so is $(\underline{R}_{\Delta}^{s})^{+}$. This implies $F_{\Delta}^{s+} \models \Lambda$. 
\end{proof}

\section{Related Literature and Further Work}
\label{sec:Related}
There is some closely related work in the existing literature which we only became aware of after completing the work reported above. 
We are grateful to the reviewers for drawing this to our attention.

Kripke semantics in which there is an accessibility relation together with an ordering on the set of worlds occurs, for instance, in~\cite{GhilardiMeloni1997,CelaniJanasa1997}.
In particular, Celani and Jansana~\cite{CelaniJanasa1997} show that the semantics for positive modal logic which they
give using this ordering has more convenient properties than earlier work by Dunn.
The stability condition used in our work
is essentialy the bimodule axiom in~\cite[p7]{GhilardiMeloni1997}. Their $\prec$ being our $\breve{H}$ and their accessibility relation, $\mid$,
 being $\conv{H} \comp R \comp \conv{H}$. Thus the (ordinary) converse of the accessibility relation in~\cite{GhilardiMeloni1997} corresponds to 
 our left converse $\leftconv{R}$ so that the adjoint pair of modalities in~\cite{GhilardiMeloni1997} would then be our $\dia$ and $\bbox$. 
However, in our system these two do not suffice to define $\bdia$ and $\Box$ as shown by our Proposition~\ref{prop:undefinable},
so the connection with our results needs further work before it can be made clear.
There is a similar situation with~\cite{GehrkeNagahashiVenema2005} where there are four independent accessibility relations and a partial order $\leqslant$. In this context our $R$ would be the relation
denoted in~\cite{GehrkeNagahashiVenema2005} by $R_\Box$ and our $\conv{H} \comp R \comp \conv{H}$ would be $R_{\dia}$.
Thus the modalities in~\cite{GehrkeNagahashiVenema2005} appear to be our $\dia$ and $\bbox$ and it is not immediately clear how our
$\bdia$ and $\Box$ fit into that framework.
Another difference between our work and these papers is that~\cite{CelaniJanasa1997} and~\cite{GehrkeNagahashiVenema2005} use a language without implication.
In fact, it is stated in~\cite[p101]{GehrkeNagahashiVenema2005} that intuitionistic implication does fit into their framework but they do not include details 
of how this is done.

Although Kripke semantics for intuitionistic modal logic generally take the form of a single set $U$ equipped with two relations $H$ and $R$ on $U$,
the approach of Ewald~\cite{Ewald1986} is somewhat different. 
Ewald uses a family of relations indexed by a poset $(\Gamma, \leqslant)$, so for each $\gamma \in \Gamma$ there is a set $T_\gamma$ and a relation
$u_\gamma \subseteq T_\gamma \times T_\gamma$. These are required to satisfy the condition that if $\gamma_1 \leqslant \gamma_2$ then
$T_{\gamma_1} \subseteq T_{\gamma_2}$ and $u_{\gamma_1} \subseteq u_{\gamma_2}$. We can however define
$U = \{(t,\gamma) \mid \gamma \in \Gamma \text{ and } t \in T_\gamma\}$ to get a single set, and then define relations $H$ and $R$ on $U$ as follows.
We put $(t_1, \gamma_1) \mathrel{H} (t_2, \gamma_2)$ iff $t_1 = t_2$ and $\gamma_1 \leqslant \gamma_2$,
and we put $(t_1, \gamma_1) \mathrel{R} (t_2, \gamma_2)$ iff $\gamma_1 = \gamma_2$ and $t_1 \mathrel{u_{\gamma_1}} t_2$.
This $R$ will not necessarily be stable in our sense, 
but it does satisfy the weaker condition that for any $X \subseteq U$ which is an $H$-set as in Definition~\ref{defn-stable},
we have $X \oplus R = X \oplus (H \comp R \comp H)$
where $\oplus$ denotes the dilation defined in the Introduction. Although this is a weaker condition than stability, the models used by
Ewald satisfy the additional constraint that for any $H$-set $X$ we have $X \oplus \breve{R} = X \oplus (H \comp \breve{R} \comp H)$.
Applying the same semantics as in our Definition~\ref{defn-semantics},
but for a language without $\coimp$,
we then have the same semantics as in~\cite{Ewald1986}
by taking the tense modalities $F, P, G, H$ to be $\dia, \bdia, \Box$ and $\bbox$ respectively.
Extending Ewald's approach to the language with $\coimp$ is then immediate but it can be shown that,
due to the weakening of stability,
the formula $\dia p \leftrightarrow \coneg \Box \neg p$ is no longer valid in all frames.
It is interesting to note that $\dia p \to \coneg \Box \neg p$ is valid in this weaker setting.

Bi-intuitionistic tense logic is studied proof-theoretically by Gor{\'e} et al.\ in~\cite{GoreBiIntAiML2010}.
This includes a discussion which obtains Ewald's semantics by identifying the two separate accessibility relations that are
used in~\cite{GoreBiIntAiML2010}. 
The use of two relations  makes the work of Gor{\'e} et al. more general than ours, a point already made in~\cite{Stell2016}, however
we contend that
the connection between $\dia$ and $\Box$ that appears explicitly in~\cite{Stell2016}, and implicitly as just noted in~\cite{Ewald1986},
is sufficiciently interesting to merit further study. 

The correspondence results mentioned here and published in~\cite{Stell2016} also appear to be closely related to earlier work. Sahlqvist theorems for positive tense logic and for intuitionistic modal logic are proved in~\cite{GehrkeNagahashiVenema2005}, and in~\cite{GhilardiMeloni1997} respectively. There is a Sahlqvist theorem for bi-intuitionistic modal mu-calculus which was established by Conradie et al~\cite{ConradieFomatatiPalmigianoSourabh2015}. Further work is needed to determine whether these results could be generalized in a straighforward way to our setting. 

In this paper we have focussed on the logical theory rather than the applications that originally motivated this work.
Within artificial intelligence the topic of spatial reasoning is of considerable practical importance~\cite{CohnRenzHbk2008}.
Spatial regions can be modelled as  subgraphs, which can also be identified with collections of pixels in images to make the connection with the 
mathematical morphology in the Introduction. Investigating the ability of the logics presented here to express spatial relations, in the sense
of~\cite{CohnRenzHbk2008}, between subgraphs is another topic for future work.
\footnote{The authors would like to thank the anonymous reviewers for their helpful and constructive comments that greatly contributed to improving the final version of the paper. The work of the first author was partially supported by JSPS KAKENHI Grant-in-Aid for Young Scientists (B) Grant Number 15K21025 and JSPS Core-to-Core Program (A. Advanced Research Networks).}

\bibliographystyle{eptcs}
\bibliography{M4Mrefs,JGSxtra}

\end{document}